\documentclass[]{article}

\usepackage{booktabs}
\usepackage{fancyvrb}
\usepackage{tikz-cd}
\usepackage{algorithm}
\usepackage{algpseudocode}
\usepackage{caption}
\usepackage{url}

\usepackage{amssymb}
\usepackage{amsmath}
\usepackage{amsthm}

\newcommand{\eat}[1]{}

\newtheorem{theorem}{Theorem}

\newtheorem{lemma}[theorem]{Lemma}

\theoremstyle{definition}
\newtheorem{example}{Example}[section]

\newcommand{\defn}{:=}
\DeclareFontFamily{U}{mathx}{\hyphenchar\font45}
\DeclareFontShape{U}{mathx}{m}{n}{
      <5> <6> <7> <8> <9> <10>
      <10.95> <12> <14.4> <17.28> <20.74> <24.88>
      mathx10
      }{}

\def\E{{\ensuremath{\mathbb E}}}

\long\def\comment#1{}

\newcommand{\dcal}{\ensuremath{\mathcal D}}

\newcommand{\ical}{\ensuremath{\mathcal I}}

\newcommand{\pcal}{\ensuremath{\mathcal P}}

\newcommand{\scal}{\ensuremath{\mathcal S}}

\usepackage{bm}

\usepackage{subcaption}

\usepackage{tikz}
\usepackage{pifont}
\newcommand{\cmark}{\ding{51}}

\usetikzlibrary{decorations.pathreplacing,calc}

\MakeRobust{\Call}

\begin{document}

\title{Conditional Cuckoo Filters}

\author{
  Daniel Ting \\
  Tableau Software \\
  1162 N 34th St \\
  Seattle, WA \\
  \texttt{dting@tableau.com} \\
  \and
  Rick Cole \\
  Tableau Software \\
  260 California Ave Ste 300 \\
  Palo Alto, CA \\
  \texttt{ricole@tableau.com}
}

\vspace{-0.1cm}
\maketitle

\begin{abstract}
	Bloom filters, cuckoo filters, and other approximate set membership sketches have a wide range  of applications. Oftentimes, expensive operations can be skipped  if an item is not in a data set. These filters provide an inexpensive, memory efficient way to test if an item is in a set and avoid unnecessary operations. 
	Existing sketches only allow membership testing for single set. However, in some applications such as join processing, the relevant set is not fixed and is determined by a set of predicates. 
	
	We propose the Conditional Cuckoo Filter, a simple modification of the cuckoo filter that allows for set membership testing given predicates on a pre-computed sketch. This filter also introduces a novel chaining technique that enables cuckoo filters to handle insertion of duplicate keys. We evaluate our methods on a join processing application and show that they significantly reduce the number of tuples that a join must process.
\end{abstract}

\vspace{-0.1cm}
\section{Introduction}
Approximate set membership data sketches such as Bloom and cuckoo filters allow a user to query whether an item $x$ belongs to a given set $\scal$, i.e. if $x \in \scal$. An item $x \in S$ in the set is always correctly classified. However, an item $x \not \in S$ has some small probability of being incorrectly classified as being in the set. In other words, the filters return no false negatives and have a small probability of returning false positives. These sketches are useful in database systems as they provide an inexpensive way to check whether an expensive operation, such as a disk access, needs to be performed.

We consider the general problem of testing set membership \emph{given predicates}. 
That is, consider a dataset $\dcal$ with each row consisting of a key $k_i$ and vector of attributes $\mathbf{a}_i$. Given a predicate $\pcal$, we wish to test if an item $x$ is in the set $\scal_{\pcal}$ of keys with attributes satisfying that predicate
\begin{align}
\scal_\pcal = \{ k: (k,a) \in \dcal \mbox{ and } \pcal(a) = \mathrm{true} \}.
\end{align}
For the purposes of this paper, we restrict ourselves to equality predicates.

We propose the Conditional Cuckoo Filter (CCF) to address this problem. 
The CCF can be seen as a simple modification of a cuckoo hash table \cite{pagh2004cuckoo} where rather than storing a key, value pair, it stores fingerprints or sketches of both. It is thus also similar to cuckoo filters \cite{fan2014cuckoo} which store only key fingerprints.
In the case of a  CCF, the value is a sketch of attribute columns. Using an attribute sketch greatly improves the functionality of the filter  at a modest cost in space. A CCF differs from both cuckoo hashes and filters in that keys may not be unique in the CCF and require techniques to handle duplicates. Tuples can share the same key but have differing attributes.
While existing cuckoo hash tables and filters only support inserting a small number of duplicated keys,
we introduce a chaining technique that allows the filter to store additional duplicates.

The CCF supports two useful operations. Given an item $x$ and predicate $\pcal$, it tests if $x \in \scal_\pcal$,in other words, if there is a matching row in the input data. Given just a predicate $\pcal$, some variations of the CCF return a cuckoo filter for the set $\scal_{\pcal}$.
Like other approximate set membership sketches it maintains the property that it cannot return false negatives.

We show that this can have significant benefits in join processing 
by reducing the number of tuples returned by intermediate scans. 
Most interestingly, it enables predicates from one table to be pushed down to scans on other tables. 
We evaluate the reduction on the real world IMDB data set.

\vspace{-0.15cm}
\section{Related work}
\emergencystretch 3em
A number of data sketches address the approximate set membership (ASM) problem including Bloom filters \cite{bloom1970bloomfilter}, 
d-left counting Bloom filters \cite{bonomi2006improved}, quotient filters \cite{bender2012quotientfilter}, and cuckoo filters \cite{fan2014cuckoo}. A number of variants \cite{putze2007cache, fan1998summary,Breslow2018mortonfilter, mitzenmacher2002compressed, wang2019vacuum, graf2019xor}, 
improve them through techniques such as cache awareness, compression, or add functionality supporting counting. These sketches significantly reduce the amount of space required in practical regimes. Theoretical work  \cite{pagh2005optimal} further reduces the asymptotic space usage to the information theoretic minimum. In all these cases, the structures only address simple set membership queries with no notion of predicates.

These structures are used in databases to speed up a number of operations. Most related to our work are the join filters used in Oracle  \cite{Das:2015:QOO:2824032.2824074, lahiri2015oracle}, Microsoft SQL Server \cite{galindo2008optimizing}, Informix XPS \cite{weininger2002efficient}, and SAP ASE \cite{sapase2013features}. Given a set of dimension tables and a large fact table, join filters construct Bloom filters when scanning the dimension tables and applying any predicates on them.
These filters effectively push down predicates on dimension tables to the fact table scan and significantly reduce its output.
Our work allows such filters to be precomputed and stored. This allows the filters to be applied to dimension tables on the build side of the join. This can result in, for example, smaller hash tables which do not spill data to disk.
\emergencystretch 0em

In databases, ASM sketches have been particularly useful for distributed join processing \cite{bratbergsengen1984hashing, Mackert1986bloomjoin, Lee2012mapreducebloom, Mullin1990semijoin} by reducing the number of tuples that  must be loaded or sent across a network.
They have also been used in join size estimation \cite{Mullin1993joinsize}
or to compute approximate join results \cite{Quoc2018approxjoin}.

ASM sketches are  extensively used in log structured merge (LSM) tree based key-value stores \cite{sears2012blsm, apachecassandra, apachehbase, rocksdb, leveldb} and in testing if query results are in a cache \cite{oracle11gresultscache}.
Outside of databases, approximate set membership sketches have a wide range of applications, particularly in networking  \cite{Broder02networkapplications, tarkoma2011theory}.

Also related are methods that sketch attribute columns. Column sketches \cite{hentschel2018columnsketch}  create small and hardware optimized data sketches that speed up scans involving a predicate. Bloom indexes in Postgres \cite{postgres} similarly use Bloom filters on rows. A scan on the small sketch locates a subset of the full data that must be read.

\section{Filters in join processing}
\label{sec:filters}
Approximate set membership sketches, and in particular Bloom filters, play an important role in reducing the costs of processing joins and semi-joins. 
By reducing the number of tuples that are processed after an initial scan, downstream  processing costs can be reduced significantly. 

Consider the following query
\small
\begin{verbatim}
    SELECT ci.*, t.title, mc.note
      FROM cast_info ci, title t, movie_companies mc
     WHERE t.id = ci.movie_id
       AND t.id = mc.movie_id
       AND ci.role_id = 4
       AND t.kind_id = 1
       AND mc.company_type_id = 2
\end{verbatim}
\normalsize

A typical query plan may build two hash tables consisting of the results of $mc$ filtered by the predicate on $mc$, followed by $t$ joined with $mc$ with predicates on $t$ and $mc$ applied.

Prebuilt Bloom or other approximate set membership filters can further restrict the set of ids to approximately those in the intersection of $ci, t, $ and $mc$ but without predicates applied.
This can substantially reduce the number of tuples that need to be added to each hash table. This is particularly useful in distributed settings where the tuples must be sent over the network.

They are also useful in the non-distributed case. 
For example, reducing the number of tuples can change a query plan from a Grace hash join that spills tuples to disk to a simple hash join that can process all tuples in memory \cite{kitsuregawa1983application}. Even if the query plan does not change, the reduction can still significantly reduce costs in building the hash table, for example, in columnar stores when each tuple incurs a row stitching cost \cite{lahiri2015oracle}.

Since approximate set membership filters do not contain any information about the predicates, they may contain far more ids than desired. 
The above query, on the IMDB dataset\footnote{We use a pre-2017 snapshot of the IMDB dataset compatible with the Join Order Benchmark}, joins the title table with the movie company and cast information tables on movie id.
This is a star join among three large tables, and hence, it is beneficial to apply as much filtering as possible when scanning all three tables. 
However, the title table contains the universe of all movie ids. Applying a prebuilt Bloom filter for the title table to the other tables is useless. However, a Bloom filter for Bollywood movies would both greatly restrict the relevant production companies and cast members.

In an ideal case, there exists a pre-built approximate set membership filter selecting only ids that match each given predicate. In that case, each hash table built consists only of ids that appear in the final result plus a limited number of false positives. In effect, the predicates for $ci$ and $t$ can be pushed down to $mc$. Our goal is exactly this, to construct a sketch which returns an approximate set membership filter for any given set of predicates.

\begin{table}
  \centering
	\begin{tabular}{c|l}
		Symbol & Meaning \\ \hline
		$h$ & Hash function \\
		$(k,\mathbf{a})$ & Query for key $k$ and attributes $\mathbf{a}$ \\
		$\ell, \ell'$ & Bucket and alternate bucket for $k$\\
		$\kappa$ & Key fingerprint \\
		$\mathbf{\alpha}$ & Attribute fingerprint vector \\
		$\beta$ & Load factor of hash table \\
		$\oplus$ & XOR operation \\
		$b$ & Number of entries per bucket \\
		$d$ & max \# of a duplicate keys in a bucket pair \\
		$H$ & Cuckoo hash table \\
		$L_{max}$ & Maximum chain length \\
		$H_\ell$, $H_{\ell,i}$ & Set of entries, or entry $i$, in bucket $\ell$  \\
		$m$ & Number of buckets in table \\
		$(k,\mathbf{a}) \in H$, $k \in H$ & CCF returns true for query 
	\end{tabular}
\caption*{Table of symbols}
\label{tbl:symbols}
\end{table}

\section{Preliminaries}
Our methods are modifications of cuckoo filters and hash tables which provide more robust capabilities for storing duplicate keys and add information about attribute values to each key. Before describing our methods, we first review cuckoo hashing techniques for those unfamiliar with them. 

Cuckoo hash tables are a form of open addressing hash table. Such a hash table is arranged as a fixed size array of entries. They avoid the overhead of storing pointers  used by separate chaining techniques to handle hash collisions, and they can often make more efficient use of cache.  
The space overhead for an open addressing hash table instead depends on the proportion of empty entries. It has low overhead when the proportion of filled entries, or load factor $\beta$, is close to 1. 
Typical collision resolution techniques such as linear probing have expected query times that grow as the load factor $\beta$ increases. This puts query speed and space efficiency in direct opposition to each other.
Linear probing has a query and insertion cost of $O(1 + 1/(1-\beta)^2)$ 

Unlike collision resolution techniques where the locations of items in the hash table are immutable after insertion,
cuckoo hashes can relocate items when there are collisions at insertion time. 
By resolving collisions at insertion time, the number of buckets that need to be probed at query time can be reduced to two while still being able to achieve a high load factor on the table. Furthermore, cuckoo hash tables have $O(1)$ amortized expected insertion time. 
Cuckoo hash tables are typically arranged in a tiered fashion so that an item is first hashed to one of $m$ candidate buckets. Each bucket contains $b$ entries in which data can be stored.

\subsection{Cuckoo insertion}
The cuckoo hashing insertion algorithm follows. 
When inserting an item, value pair $(k, v)$,
the item is hashed to two possible buckets $\ell, \ell'$.
If $k$ is already in one of the buckets, then the value is simply updated.
Otherwise, if either bucket is not full, the pair is simply added to a non-full bucket with $\ell$ being preferred over $\ell'$.
If both are full then it picks a random pair $(k',v')$ from the two buckets. It "kicks out" that pair and replaces it with $(k,v)$. 
The pair $(k',v')$ is then reinserted into the sketch at its alternate bucket. 
The process of kicking out and reinserting is performed up to some maximum $\mathrm{MaxKicks}$ number of times. The process reaches this maximum then the table is resized and all pairs are reinserted.
To query a cuckoo  hash table, only the two possible buckets need to be examined for the key.

\subsection{Cuckoo filters}
Cuckoo filters can be seen as a particular form of cuckoo hash table. 
They are used for ASM queries rather than for key-value queries. 
There are two primary differences from typical cuckoo hash tables. 
First, only a small fingerprint $\kappa$ of the key $k$ is stored in the table. 
It does not store the full key nor any value associated with the key. 
Second, it uses partial-key cuckoo hashing
where the alternate bucket $\ell' = \ell \oplus h(\kappa)$ is determined only by the bucket $\ell$ and the fingerprint.
This allows the alternate bucket to be computed using the limited information stored in the sketch rather than requiring the discarded full key. Here $\oplus$ is the XOR operation and $h$ is a hash function.

To check if a key $k$ exists, the $\leq 2b$ items in its 2 buckets $\ell, \ell'$ are checked for a matching fingerprint. Trivially, if $k$ was previously inserted into the filter then the filter will find a match, so that there are no false negatives. If $k$ was not previously inserted, then there is probability $D \times 2^{-|\kappa|}$ that there is a matching fingerprint due to random chance where $D \leq 2b$ is the number of filled entries in the two buckets and $|\kappa|$ is the size of the key fingerprint in bits.
The False Positive Rate (FPR) of the filter is thus $\rho = 2^{-|\kappa|} \E D$.

For a typical setting where the number of entries per bucket $b=4$, an optimally sized cuckoo filter requires approximately
$(\log_2 1/\rho + 3) / \beta$ bits per item to achieve a desired FPR of $\rho$ where $\beta$ is the load factor of the table. According to \cite{fan2014cuckoo}, an optimally sized filter empirically achieves $\beta \approx 95\%$ when $b=4$. Compared to Bloom filters which require $\approx 1.44 \log_2 1/\rho$ bits per item to achieve the same FPR, an optimally sized cuckoo filter requires fewer bits per item when the desired FPR $\rho < 0.35\%$. 

In order to further reduce the number of bits per item needed to achieve a target FPR, the entries in the bucket can be sorted. This reduces the entropy of the bucket and allows for a more efficient encoding. This can be done efficiently if only 4-bit prefixes of the fingerprints are sorted. In this case, the bits per item needed is reduced to $(\log_2 1/\rho + 2) / \alpha$.
This reduction allows cuckoo filters to use fewer bits than a Bloom filter when the target FPR $\rho < 2.5\%$.

\subsection{Multisets}
Cuckoo filters have limited support for multisets and duplicate keys. They can be extended either by adding a counter to each entry or by inserting an additional copy of a key fingerprint. We focus on the latter as we  additionally store unique attributes with each duplicated key.
This makes cuckoo filters more flexible than Bloom filters as they support deletions by removing a copy of a key fingerprint.
However, there is a cap of $2b$ copies that can be inserted since a key can only probe $2b$ entries in its 2 buckets. Furthermore,  there are no theoretical guarantees that a high load factor can be obtained. Empirically, we see they cannot. Figure \ref{fig:multiset dupes} shows that the load factor decreases dramatically when there are duplicate keys, and insertions into the filter can fail almost immediately when the distribution of duplicated keys is highly skewed.

\section{Conditional cuckoo filters}
We now introduce the Conditional Cuckoo Filter (CCF).
These support set membership queries with equality predicates.
Like a cuckoo filter, a CCF is based on cuckoo hashing and saves space by using fingerprint or sketches. 
By storing sketches for both key and attribute values rather than just the key, they provide a space efficient structure that supports predicates. This can be much more efficient than the alternative which stores a separate filter for each combination of predicate values. Such a strategy would grow exponentially in size.

There are two main problems in adding back this attribute information. 
First, how can a set of attributes be summarized and stored in a small amount of space? 
Second, how can the cuckoo hash table deal with non-unique keys, especially when the key distribution is highly skewed?

We propose and evaluate several solutions to these problems. 
To solve the former, we introduce the extremely simple, but novel, idea  of sketching attributes and provide three ways of sketching attributes: using a vector of fingerprints, a Bloom filter, or a mixture of the two. For all of these, the primary novel extensions to cuckoo filters and hashing are in how duplicate keys with unique fingerprints are handled as cuckoo filters quickly fail in the presence of duplicates. The methods we propose employ either a chaining mechanism which allows the CCF to use more buckets as more duplicates are encountered or a method to switch from fingerprint vectors to Bloom filters.

\subsection{Attribute fingerprint vectors}
The simplest method for sketching a vector of attributes  hashes each attribute value into a small number of bits $s$, say 4 or 8, to construct a vector of attribute fingerprints.
Despite their extremely small size, these attribute sketches can  be effective in applications. 
For example, in join processing
the probability of a false match does not need to be extremely low to be effective. 
The expected output size after applying an equality predicate to an intermediate scan is
\begin{align*}
	\E M_{output} &= M_{true} + \E\, \mathit{FPR} *  (M_{original} - M_{true}) \\
	&< M_{true} + \E\, \mathit{FPR} * M_{original}.
\end{align*}
Thus, when the number of true matching tuples is a small percentage of the output tuples, $M_{true} / M_{output} \approx 0$, a relatively poor FPR of just $10\%$  reduces the number of tuples produced by a factor of nearly $10$.
Furthermore, if more than one predicate is applied, the reduction can be multiplicatively amplified. 
Similarly, there is little reason to target a FPR much smaller than the ratio $M_{true} / M_{original}$, as this unavoidable cost becomes the dominant cost in processing.

The resulting CCF will thus have a very low FPR when a key is absent from the set but allows for a higher FPR if the key exists but there is no matching attribute.
Similar to key where we represent a key $k$ with its Greek counterpart $\kappa$, we use $\mathbf{a}$ to denote attribute vectors and
$\mathbf{\alpha}$ to denote their attribute fingerprint vector. When using attribute fingerprint vectors, the underlying cuckoo hash table must be modified to handle duplicated keys with unique attribute fingerprints.

\subsection{Bloom filter attribute sketches}
\label{sec:bloom attributes}
A second choice represents attributes with a Bloom filter. Each (attribute name, value) pair is inserted into a small Bloom filter. 
The  resulting sketch is simply a cuckoo filter with an added Bloom filter for each entry. Algorithm \ref{alg:query} summarizes the procedure for querying the filter. The only difference from a regular cuckoo filter query is the additional check to verify if the attribute matches. A CCF using Bloom attribute sketches can also support queries that only specify a predicate and not a key. It returns a cuckoo filter which a downstream process can use to check the existence of a key. To do this, simply erase all entries where the predicate does not match and return the resulting array of key fingerprints. This is summarized in algorithm \ref{alg:predicate query}.

\begin{algorithm}[H]
	\begin{algorithmic}
		\State $(\kappa, \ell) \gets h(k)$
		\State $\mathbf{\alpha} \gets h_A(\mathbf{a})$ 
		\State $\ell' \gets \ell \oplus h(\kappa)$
		\For{$(\kappa', \mathbf{\alpha}') \in H_\ell \cup H_{\ell'}$}
		\If{$\kappa = \kappa'$ and $\mathbf{a}$ matches $\mathbf{\alpha}'$}
		\State \Return True
		\EndIf
		\EndFor
		\State \Return False
	\end{algorithmic}
	\caption{Query($k, \mathbf{a}$)}
	\label{alg:query}
\end{algorithm}

\begin{algorithm}[H]
	\begin{algorithmic}
		\State $m, b \gets \mathrm{Dimension}(H)$
		\State $H' \gets$ new CuckooFilter($m,b$)
		\For{$\ell \gets 1, \ldots m$}
		\For{$i \gets 1,\ldots, b$}
		\State $(\kappa, \mathbf{\alpha}) \gets H_{\ell, i}$
		\If{$\mathbf{a}$ matches $\mathbf{\alpha}$}
		\State $H'_{\ell,i} = \kappa$
		\EndIf
		\EndFor
		\EndFor
		\State \Return H'
	\end{algorithmic}
	\caption{PredicateQuery($\mathbf{a}$)}
	\label{alg:predicate query}
\end{algorithm}

Using a Bloom filter attribute sketch has mixed effects on the required size of the sketch and the FPR.
First, note that the occupied entries in the sketch are exactly the same as those of a cuckoo filter. Thus, appropriately sized filters are theoretically guaranteed \cite{eppstein2016cuckoo} to obtain high load factors with high probabilty. For our other methods, we only have empirical results showing high load factors are obtained. However, a Bloom filter is less bit efficient than a fingerprint vector. An optimized Bloom filter requires $\approx 1.44 \log_2 (1/\rho)$ bits per attribute to achieve an FPR of $\rho$ versus $\log_2(1/\rho)$ for a fingerprint vector. This inefficiency is exacerbated since it is not possible to choose optimal parameters for the Bloom filter. The optimal choice for the number of hash functions to use depends on the number of distinct (attribute, value) pairs that will be added to the filter. These are not known in advance and can vary greatly for different keys.
Second, when multiple rows of data share the same key but have different attribute vectors, a Bloom filter attribute sketch does not encode which attribute values occur together in the same row. If row 1 has attributes $(a_1, a_2)$ and row 2 has $(a_1',a_2')$, then given a predicate $A_1 = a_1 \land A_2 = a_2'$, there are no matching rows for the predicate, but a Bloom filter attribute sketch is guaranteed to return a false positive.
This would remain true if the Bloom filter were replaced with any other ASM sketch.
On the other hand, if all queries contain only a single equality predicate, then the fingerprint vector can unnecessarily store a single attribute value $A_1=a_1$ multiple times when the second attribute $A_2$ is varying.  A Bloom filter only encodes it once.

\emergencystretch 3em
\section{Multiset representations }
\emergencystretch 0em
Although one advantage of attribute fingerprint vectors is the ability to store  co-occurence information, each unique fingerprint vectors must still occupy distinct entries in the CCF. The ability to handle a potentially large number of duplicates is an important ability that normal cuckoo hash tables lack. A key's 2 buckets contain at most $2b$ entries, and inserting any more copies of a key is guaranteed to fail. This is problematic as many data sets have highly skewed distributions for the key.

We present two strategies to address this and maintain a no false negative guarantee. One converts attribute fingerprint vectors to  Bloom filters when too many duplicates are encountered. The second is a form of chaining that allows a key to utilize more than 2 buckets.

In both cases, we allow a maximum of $d$ duplicated key fingerprints per bucket pair.
If an attempted insertion for key $k$ is on a bucket pair  already containing $d$ copies of the fingerprint $\kappa$,
then either the $d$ copies are converted to a Bloom filter attribute sketch or additional bucket pairs are considered via the chaining procedure.

\subsection{Bloom filter conversion}
Consider a key $k$ and attribute vector $\mathbf{a}$ that must be inserted in the pair of buckets $\ell, \ell'$.
Let $|\kappa|, |\mathbf{\alpha}|$ denote the size of the key fingerprint and attribute fingerprint vectors and
$\#\mathbf{\alpha}$ denote the number of attributes.
Suppose there are already $d$ copies of the fingerprint $\kappa$ in the $2b$ entries of $\ell,\ell'$.
Bloom filter conversion takes the $d |\mathbf{\alpha}|$ bits that are currently used to store $d$ fingerprint vectors and constructs a single Bloom filter in their stead. Each entry in the sketch also requires an additional bit to track whether it contains a Bloom filter attribute sketch or a fingerprint vector.

This conversion operation has the advantage that it can never fail.  However, it adds complexity in storing a Bloom filter among $d$ entries and maintaining it whenever a bucket's entry is kicked into the alternate bucket. It has the same advantages and disadvantages outlined in section \ref{sec:bloom attributes} when directly using a Bloom filter attribute sketch but with two main differences. 
It has a further disadvantage in that hash collisions can be introduced both from hashing attribute values into fingerprints and from inserting fingerprints into the Bloom filter. Directly using a Bloom filter only introduces collisions from the latter. It has an advantage in that
the Bloom filter parameters can be chosen more easily since the minimum number of duplicates is $d$. The Bloom filter parameters do not need to be optimized to also handle rows with a unique key. 

The storage of the entries can be further optimized to avoid storing the same key fingerprint multiple times. Instead, each bucket can store a single copy of the key fingerprint along with the number of entries the Bloom filter attribute sketch occupies in that bucket.
If the other entries sharing the same key fingerprint are stored contiguously, then the Bloom filter can be successfully reconstructed. 
In this case, the required number of bits to store the fingerprint and counts is $2(|\kappa| + \lceil \log_2 d \rceil)$ and the number of bits in $d$ entries is $d (|\kappa| + |\alpha|  + 1)$. 

We choose the number of hash functions to be approximately the optimal number assuming there are $(d+1) \cdot \#\mathbf{\alpha}$ unique attribute name, value pairs added to the filter. Let $|B|$ be the number of bits available to the Bloom filter.
\begin{align}
	\mbox{\# hashes} &\approx \frac{|B|}{(d+1) \cdot \#\mathbf{\alpha} } \log 2\\
	&\approx \frac{|\mathbf{\alpha}|}{\#\mathbf{\alpha}}\frac{d}{d+1} \log 2 \quad \mbox{if $|\mathbf{\alpha}| \gg |\kappa|$}
\end{align}

Algorithm \ref{alg:bloom conversion} summarizes the method to convert attribute fingerprint vectors to a Bloom filter.

\begin{algorithm}
	\caption{BloomConversion($H, \ell, \ell', \kappa, \mathbf{\alpha}$)}
	\label{alg:bloom conversion}
	\begin{algorithmic}
		\State Sort entries in $H_\ell$ and $H_{\ell'}$ by fingerprint
		\State $r_i \gets \#$ entries with fingerprint $\kappa$ in $H_i$ for $i =\ell, \ell'$.
		\State $s \gets $ Size of single entry
		\State $numHash \gets \frac{|\mathbf{\alpha}|}{\#\mathbf{\alpha}}\frac{d}{d+1} \log 2$
		\State $totalBits \gets d s - 2(|\kappa| + \lceil \log_2 d \rceil)$
		\State $B \gets Bloom(numHash, totalBits)$
		\For{$(\kappa', \mathbf{\alpha'}) \in H_\ell \cup H_{\ell'} \cup \{(\kappa, \mathbf{\alpha})\}$}
		\If{$\kappa' = \kappa$}
		\For{$j \gets 1,\ldots, \#\mathbf{\alpha}'$}
		\State Insert $(j,\alpha_j)$ into $B$
		\EndFor						
		\EndIf
		\EndFor
		\State $\underline{\ell}, \overline{\ell} \gets \min\{\ell,\ell'\} , \max\{\ell,\ell'\}$
		\State $\mathit{bits}_{\underline{\ell}} \gets r_{\underline{\ell}} s - |\kappa| + \lceil \log_2 d \rceil$
		\State Pack $(\kappa, r_{\underline{\ell}}, B_{1,\ldots, \mathit{bits}_{\underline{\ell}}})$ into the $r_{\underline{\ell}}$ entries in $H_{\underline{\ell}}$
		\State Pack $(\kappa, r_{\overline{\ell}}, B_{\mathit{bits}_{\underline{\ell}}+1, \ldots, numBits})$ into the $r_{\overline{\ell}}$ entries in $H_{\overline{\ell}}$
	\end{algorithmic}
\end{algorithm}

\subsection{Chaining}
Chaining introduces additional bucket pairs whenever an insertion 
would violate the constraint that at most $d$ copies of a key fingerprint $\kappa$
 are in a bucket pair $\ell, \ell'$. 
A second bucket pair is determined by hashing both the bucket pair and fingerprint,
$\tilde{\ell} \defn h(\min\{\ell, \ell'\}, \kappa)$. Unlike the XOR operation for bucket pairs, determining the second bucket pair is a one way operation. The second bucket pair is computable from the first, but not vice-versa. The min can also be replaced by some other symmetric function.
This may be repeated to generate a chain of bucket pairs. 
We cap the number of bucket pairs that are generated by a constant $L_{max}$. The procedure is illustrated in figure \ref{fig:chaining}. The algorithm is summarized in algorithm \ref{alg:insert}.

\begin{figure}[H]
  \centering
	\includegraphics[width=.8\textwidth]{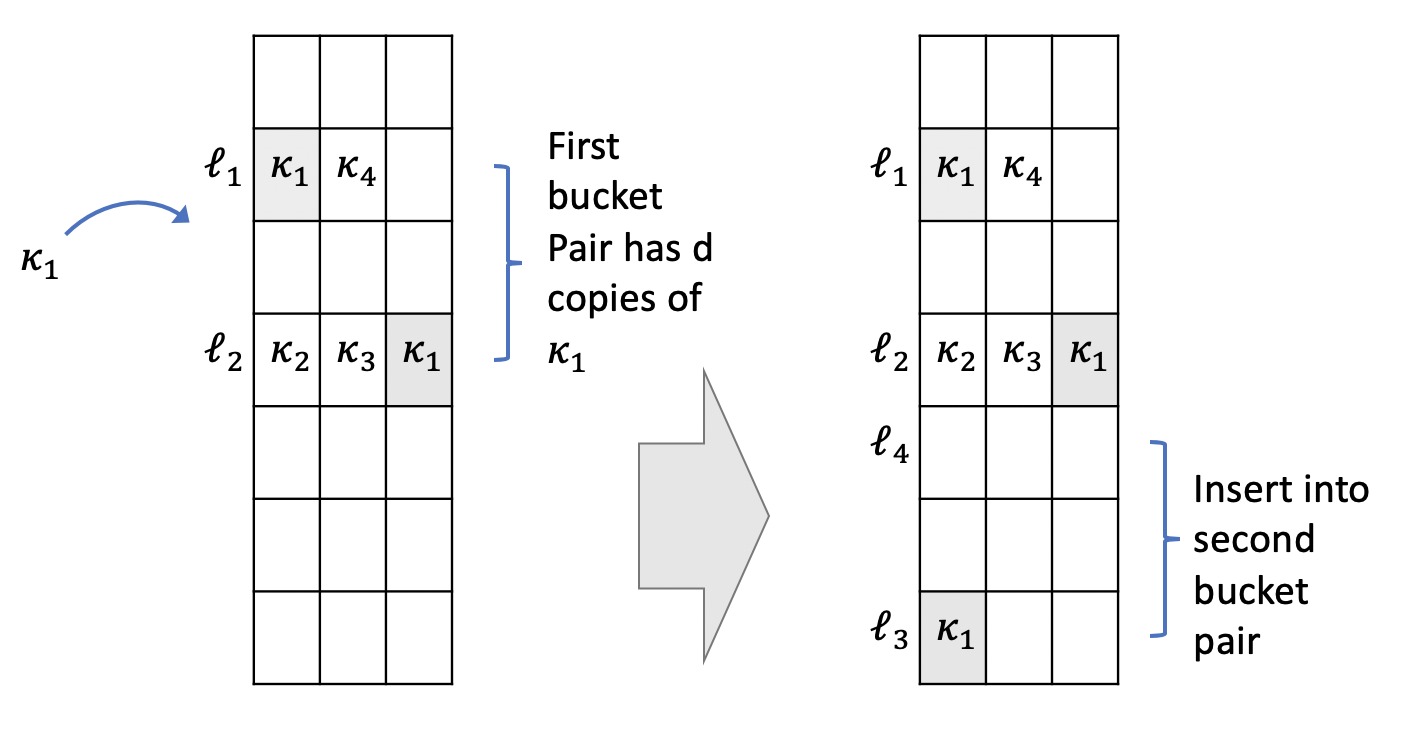}
	\caption{Illustration of chaining procedure }
	\label{fig:chaining}
\end{figure}

\begin{algorithm}[H]
	\begin{algorithmic}
		\State $\ell' \gets \ell \oplus h(\kappa)$
		\If{$(\kappa, \mathbf{\alpha}) \in H_\ell \cup H_{\ell'}$}
		\State \Return Success
		\ElsIf{$| \{(\kappa', \alpha') \in H_\ell \cup H_{\ell'} : \kappa' = \kappa \}| = maxDupes$}
		\State $\tilde{\ell} \gets h(\min\{\ell, \ell'\}, \kappa)$
		\State ChainInsert($\tilde{\ell}, \kappa, \mathbf{\alpha}, L_{max}-1$)
		\ElsIf{$H_\ell$ is not full}
		\State Insert $(\kappa, \mathbf{\alpha})$ into $H_{\ell}$
		\State \Return Success
		\EndIf
		
		\For{$i \gets 1$ to $\mathrm{MaxKicks}$}
		\If{$H_{\ell'}$ is not full}
		\State Insert $(\kappa, \mathbf{\alpha})$ into $H_{\ell'}$
		\State \Return Success
		\EndIf
		\State Pick random entry $i \leq b$
		\State Swap$(\kappa, \mathbf{\alpha})$ with $H_{\ell'}[i]$
		\State $\ell' \gets \ell' \oplus h(\kappa)$
		\EndFor	
		\State \Return Terminated
	\end{algorithmic}
	\caption{ChainInsert($\ell, \kappa, \mathbf{\alpha}, L_{max}$)}
	\label{alg:insert}
\end{algorithm}

When querying for a key $k$ and predicate $\pcal$, the next bucket pair is checked only if there are exactly $d$ entries with the given key fingerprint. In the case where $L_{max}$ bucket pairs are checked and the last bucket pair contains $d$ copies of the given key fingerprint, the query will return true  regardless of the predicate. Otherwise, it returns true only if the predicate finds a match in one of the attribute sketches for that key. 
The algorithm is summarized in algorithm \ref{alg:chainquery}.
The correctness of this is proven by the following lemmas.

\begin{algorithm}
	\begin{algorithmic}
		\For{$i \gets 1, \ldots, L_{max}$}
			\State $\ell' \gets \ell \oplus h(\kappa)$
			\If{$(\kappa, \mathbf{\alpha}) \in H_\ell \cup H_{\ell'}$}
				\State \Return True
			\ElsIf{$| \{\kappa' \in H_\ell \cup H_{\ell'} : \kappa' = \kappa \}| = d$}
				\State $\ell \gets h(\min\{\ell, \ell'\}, \kappa)$
			\Else
				\State \Return False
			\EndIf
		\EndFor
			\end{algorithmic}
	\caption{ChainQuery($\ell, \kappa, \mathbf{\alpha}, L_{max}$)}
	\label{alg:chainquery}
\end{algorithm}

\begin{lemma}
	Given any key fingerprint $\kappa$ and bucket $\ell$, let $\ell' = \ell \oplus h(\kappa)$
	be the alternative bucket. The total number of copies of $\kappa$ in buckets $\ell, \ell'$ only increases as items are inserted into a conditional cuckoo filter with chaining, and the total number is capped by the parameter $d$.
\end{lemma}
\begin{proof}
	We can prove this by induction. This trivially holds in the base case where the conditional cuckoo filter is empty.
	Upon insertion of a new item, the only way the number of copies of $\kappa$ can decrease is if an entry containing $\kappa$ is kicked out by a cuckoo kick operation. Assume WLOG that $\kappa$ was in bucket $\ell$.
	If it is kicked out by a row with key fingerprint not equal to $\kappa$ then the number of copies of $\kappa$ in $\ell, \ell'$ must be strictly less than $d$. Therefore, it must be reinserted into bucket $\ell'$ and the invariance holds.
	If it is kicked out by another row with key $\kappa$, then if there were $< d$ copies of $\kappa$ in $\ell, \ell'$ before, there are still $<d$ copies and it must similarly be reinserted into $\ell'$. This increases the count of copies of $\kappa$ but it cannot exceed $d$.
	If there were $d$ copies before, then there are still $d$ copies after the kick operation. The chaining operation ensures that $\kappa$ is not reinserted into the pair $\ell, \ell'$, so the cap is preserved.
	Consider a possible insertion for fingerprint $\kappa$ into either $\ell$ or $\ell'$. WLOG assume it is a possible insertion to $\ell'$.
	The alternative bucket is $\ell' \oplus h(\kappa) = \ell$. An insertion will check if the cap would be preserved and move to the next bucket pair if it is not. 
\end{proof}

\begin{lemma}
	\label{lem:d-chain}
	Let $k, \mathbf{a}$ be a key, attribute vector pair and $C$ a conditional cuckoo filter with chaining. There is a fixed sequence of buckets $\ell_1, \ell_2, ... \ell_n$ with $n \leq L_{max}$  that the sketched entry $\kappa, \mathbf{\alpha}$ corresponding to this pair can be inserted into. Furthermore, if it is into $\ell_i$, then all pairs $\ell_{2j}, \ell_{2j-1}$ with $2j-1 < i$ contain $d$ copies of $\kappa$.
\end{lemma}
\begin{proof}
	The sequence is given by $\ell_1 = h_{\ell}(k)$, $\ell_{2j} = \ell_{2j-1} \oplus h(\kappa)$ and $\ell_{2j+1} = h_b(\min\{\ell_{2j-1}, \ell_{2j}\}, \kappa)$. This recursion is defined by only $k$ for the first bucket and only by $\kappa$ and the previous bucket for all other buckets. It ends if there is a cycle or the maximum chain length $L_{max}$ is reached.
	The other invariant holds by induction. It trivially holds for the empty filter.
	Consider $i'$ even. Based on the lemma above, an entry containing $\kappa$ in $\ell_{i'}$ can only move to $\ell_{i' + \Delta}$ 
	if there were already $d$ copies of $\kappa$ in $\ell_{i'-1}, \ell_{i'}$. This hold regardless of whether the move is due to a simple insertion or from being first kicked out. Thus, the invariant is preserved. Similarly, it is preserved when $i'$ is odd.
\end{proof}

\begin{theorem}
	The conditional cuckoo filter with chaining returns no false negatives.
\end{theorem}
\begin{proof}
	If a key $k$ and predicate $\pcal$ have some matching row $k, \mathbf{a}$, then that row must either be inserted into some bucket in the sketch or discarded because the maximum chain length is exceeded. The lemma above ensures that a query either finds the corresponding entry or reaches the maximum chain length. Both conditions return true.
\end{proof}

Note that generating alternative bucket pairs may result creating a cycle of bucket pairs. In this case, the insertion procedure will not fail but will not generate $L_{max}$ unique bucket pairs. To further improve the chaining procedure, such cycles can be detected and the chain can be extended. We can do this using Floyd's cycle detection algorithm.

Like Bloom filter attribute sketches, the chaining method can also support predicate only queries. However it cannot simply erase entries with non-matching attribute values. This could introduce gaps in a chain where some bucket pair does not contain $d$ copies of a fingerprint. A query could improperly stop probing bucket pairs early and yield false negatives. Instead, the sketch must keep the key fingerprint and use an additional bit to mark the entry as non-matching.

\section{False Positive Rates}
The accuracy of CCF's and other approximate set membership sketches can be measured by the false positive rate (FPR). Unlike regular cuckoo filters, the FPR for CCFs is not a constant.
Queries can result in false positives due to  spurious matches on the key fingerprint, on the attribute sketch, or both.
Because of this, the FPR depends on the distribution of the underlying data and the query itself. 
We provide some simple bounds on the FPR expressed in these relevant quantities.

\subsection{Key only queries}
For a standard set membership query only a key with no predicates, the CCF has a FPR similar to a regular cuckoo filter.
When using Bloom filter attribute sketches, the CCF is, in fact, identical to a cuckoo filter when the attribute sketches are ignored. Thus, they have exactly the same FPR as cuckoo filters on key only queries. We show that all the variations of the CCF are governed by a bound of the same form for key only queries. 

The typical bound on the FPR on a cuckoo filter is given by a union bound. A key $k$ that was not inserted into the sketch generates a random key fingerprint $\kappa$ with $|\kappa|$ bits. The probability that $\kappa$ randomly matches any given entry in the sketch is $2^{-|\kappa|}$. Summing over the entries that can be probed by a query gives a union bound on the FPR.
A slight refinement of the typical bound for the FPR sums this probability for a spurious match over the number of non-empty entries $D$ in the key's bucket pair rather than all $2b$ entries.
This gives a bound on the FPR for key only queries using Bloom attribute sketches as
\begin{align}
	FPR^{key} &\leq \E D 2^{-|\kappa|}
\end{align}

For a CCF using Bloom conversion, only entries with distinct key fingerprints in the bucket pair need to be counted. Letting $D$ represent this distinct count gives the same bound as above.
We note that although the form of the bound is the same, this does not imply that a similarly sized CCF using Bloom conversion yields a smaller FPR. To avoid insertion failures, a sketch using Bloom conversion must contain more entries as it stores duplicates.

For a CCF using chaining, lemma \ref{lem:d-chain} shows that the chain is irrelevant for key only queries. Only the first bucket pair must be checked. 
Thus, the formula for the FPR is also the same as above. This result is of interest since although insertions can probe up to $2 L_{max}$ buckets, there is no penalty for probing more buckets 
at query time.

\subsection{Key and predicate queries}
CCF's are distinguished by their ability to answer queries for a key and predicates.
Consider a query for a key $k$ and equality predicates  $A_i=a_i$ for attributes $i=\ical$. Further suppose
there are no matches on the full data so that a positive return value is a false positive.
The probability that a CCF returns true can be decomposed as
\begin{align}
	p( (k,\pcal) \in H) &= p(k \in H) p( \pcal \in H[k] | k \in H) 
\end{align}
where $k \in H$ or $(k,\pcal) \in H$ denotes the event that the filter returns true for the given query,
$H[k]$ denotes the entries that a query involving $k$ would probe and which contain the key fingerprint $\kappa$, and $\pcal \in H[k]$ denotes that there is a match for the predicate among those entries.

Consider the following cases
\begin{itemize}
	\item The key $k$ is not in the data
	\item The key $k$ is in the data, but there is no match for the predicate 
\end{itemize}
In the former case, the FPR is trivially bounded by $p(k \in H)$ which is upper bounded in the previous subsection. An upper bound on the FPR of $\leq 5\%$ can be achieved with a key fingerprint size of $8$ and 6 buckets per entry. 
This bound holds for all variations of the CCF.

In the latter case, $p(k \in H) = 1$ since the key is in the data and there are no false negatives. A false positive occurs if there is a spurious match for  the predicate $\pcal$
among the entries in $H[k]$.
The FPR is thus the probability the predicate matches on an attribute sketch, $p( \pcal \in H[k] | k \in H) = p(\pcal \in H[k])$. 

The FPR differs depending on the attribute sketch used.
For Bloom filter attribute sketches, the FPR is 
\begin{align}
	p_{Bloom}(\pcal \in H[k]) &= \rho_k^v
\end{align}
 where $\rho_k$ is the FPR of the Bloom filter for key $k$ and $v$ is the number of attribute values which were never inserted into the Bloom filter. 
 
Note that when the predicate tests for the co-occurrence of multiple attributes, 
the FPR can be 1 as described in section \ref{sec:bloom attributes}.  
The standard formula for the FPR for a Bloom filter is $\rho \approx (1-(1-h/s)^{n})^h \approx(1-exp(-hn/s))^h$ where $h$ is the number of hash functions used by the Bloom filter, $s$ is the number of bits, and $n$ is the number of unique attribute name, value pairs that are added to the Bloom filter. In this case where the Bloom filters are small, this approximation is an underestimate of the FPR \cite{bose2008false}.

When using attribute fingerprints, the probability of a spurious match on one entry has a similar form, 
$p(\pcal \in H[k]_i) = \tilde{\rho}_k^{\tilde{v}}$ where $H[k]_i$ is
 the $i^{th}$ entry that is checked with fingerprint $\kappa$. This entry corresponds to some input row that was added to the filter.
Here, $\tilde{V}$ is the number of attributes in the predicate which do not match the underlying data row's attributes and
$\tilde{\rho} = 2^{-|\alpha|}$ where $|\alpha|$ is the size of the attribute fingerprint used.
In this case there are a maximum number of $d L_{max}$ entries checked, so the overall FPR is bounded by
\begin{align}
	p_{chained}(\pcal \in H[k]) &\leq d L_{max} \E 2^{-|\alpha| \tilde{V}}.
\end{align}
Similar to the key-only query case, the FPR does not depend on the total number of entries that are probed at insertion time but only on a smaller subset containing the key fingerprint $\kappa$.

For Bloom conversion, the FPR depends on whether the entry for the key $k$ has been converted to a Bloom filter.

Although the formulae are upper bounds on the expected FPR, figure \ref{fig:predicted fpr} shows they are reasonably good predictors of the actual FPR. 

\begin{figure}[H]
\vspace{-0.2cm}
  \centering
	\includegraphics[width=.8\textwidth]{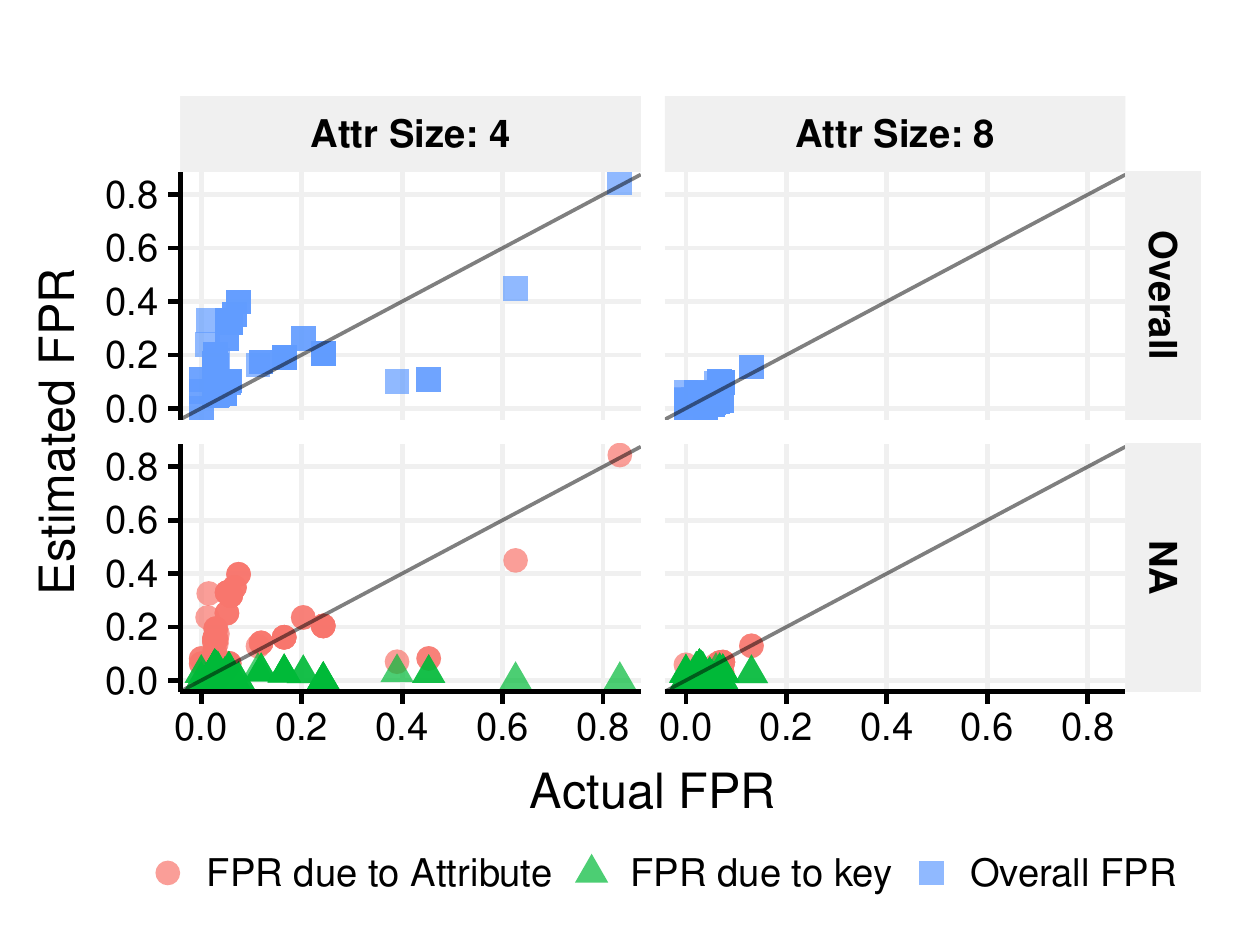}
	\caption{The bounds on the expected FPR are good predictors of the actual FPR when using attribute fingerprints. }
	\label{fig:predicted fpr}
\vspace{-0.2cm}
\end{figure}

\section{Size and parameter choice}
\label{sec:param choice}
Conditional cuckoo filters have more parameters than regular cuckoo filters. Cuckoo filters only require setting the number of buckets $m$ and the number of entries per bucket $b$. CCF's can additionally require setting the maximum number $d$ of duplicates per bucket pair, $L_{max}$ the maximum chain length, and any additional parameters required by the attribute sketch.

We derive an upper bound on the number of non-empty entries and show through experiments that the attainable load factor is a constant that is insensitive to the underlying data. Together these can be used to size the sketch.
These bounds and constants depend on the parameters $d$ and $L_{max}$ which affect the number of duplicates stored in the sketch.

Denote the total number of distinct keys by $n_k$ and the number of non-zero entries in a CCF by $Z'$.
Since a CCF with Bloom attribute sketches has the same non-empty entries as a regular cuckoo filter, the number of non-zero entries can be upper bounded by $n_k$.
For other cases, let $r_k$ be the number of duplicates for key $k$ that have distinct attribute values. Bloom filter conversion will allocate a maximum of $\max\{d, r_k\}$ entries for that key. Let $A = r_X$ for a randomly chosen key $X$. Then the expected number of used entries  is bounded by $\E Z' \leq n_k \E \min\{A, d\}$.
Similarly, a CCF with chaining uses at most $d L_{max}$ entries for a single key, so the expected number of entries used is bounded by 
$\E Z' \leq n_k \E \min\{A, d L_{max}\}$. These sizes are summarized in table \ref{tbl:ccf types}. Figure  \ref{fig:filter sizing} shows that the bound on the number of entries needed closely matches the actual number needed.

\begin{table}[H]
  \centering
	\begin{tabular}{l|ccc|l}
		& \multicolumn{3}{c|}{Queries} &  \# non-empty entries\\ 
		Filter & $k$ & $(k, \pcal)$ & $\pcal$ & (upper bound) \\ \hline
		Cuckoo filter & \cmark &  &  & $n_k$ \\
		CCF w/ Bloom & \cmark & \cmark & \cmark & $n_k$ \\
		CCF w/ conversion & \cmark & \cmark & \cmark & $n_k \,  \E \max\{A, d\}$\\
		CCF w/ chaining & \cmark & \cmark &  &  $n_k\, \E \max\{A, d L_{max}\}$ \\
		\end{tabular}
	\caption{Supported queries and sizing for different conditional cuckoo filters. $n_k$ is the total number of distinct keys and $A$ is the number of distinct attribute vectors associated with a randomly chosen $k$.}
	\label{tbl:ccf types}
\end{table}

\begin{figure}[H]
\vspace{-0.2cm}
  \centering
	\includegraphics[width=.8\textwidth]{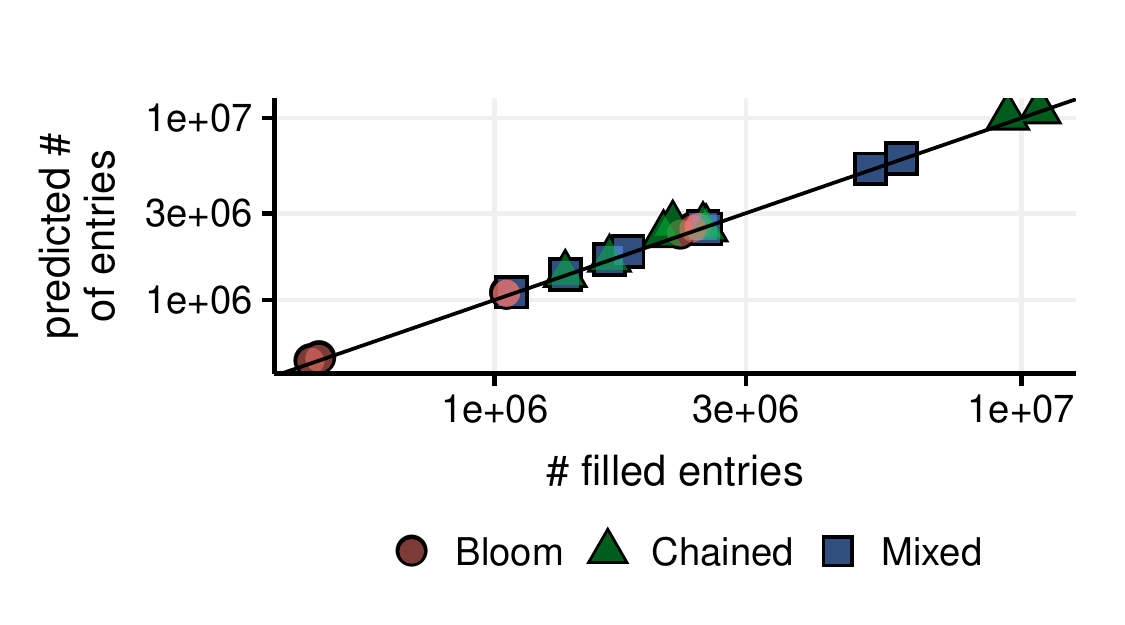}
	\caption{The predicted number of entries needed for each filter closely match the actual number for the workload.}
	\label{fig:filter sizing}
\vspace{-0.2cm}
\end{figure}

From experiments on chaining, we find a reasonable rule of thumb for setting the number of entries per bucket $b$ is to take $b \approx 2d$. This way, at least 4 keys can be stored in a bucket pair $\ell,\ell'$. Typically more keys can be stored since only 1 key fingerprint will use that particular pair of bucket locations. This provides relatively high load factors while ensuring buckets are not costly to scan. Figure \ref{fig:multiset dupes} shows that a setting of $b=4$ that is typical for cuckoo filters achieves a load factor of around $\beta \approx 75\%$ regardless of the number of duplicate keys. A slightly larger value of $b=6$ achieves a load factor close to $\beta \approx 87\%$ even when there are many duplicates. 
An appropriate estimate for the required size for a CCF is $m \cdot b \approx \E Z' / \beta$.

When the sketch is optimally sized, figure \ref{fig:multiset fpr/maxdupes} shows that lower settings for $d$ tend to achieve better use of bits. This is primarily due to smaller values of $d$ yielding higher load factors. Given a fixed setting for the number of bins $m$ and entries per bin $b$, we found the best $d$ is the largest under which all insertions pass. We chose $d=3$ which provides small bucket sizes and a good load factor. 

\subsection{Attribute sketch parameters}
Figure \ref{fig:overall RF} shows the performance of different CCF's under various parameter choices.
Generally, we also found increasing the attribute sketch size more beneficial than increasing the key fingerprint size. Figure \ref{fig:predicted fpr}  shows that at small attribute sketch sizes, the false positives are primarily due to bad matches on the attribute sketches. In this case, a false positive $k, \mathbf{a}$ is attributed to a key if the key is not in the sketch, in other words, $k \not \in H$.
Otherwise, it is easily attributed to the attribute.
We also found small values for the number of hash functions used by Bloom filters to be preferable as it does not become filled with ones too quickly.

\section{Additional optimizations}
\label{sec:optimizations}
The chaining method lends itself to several optimizations. 

{\bf Storage:} 
When using attribute fingerprint vectors, the CCF is an open addressing hash table, and can be directly stored as such. Furthermore, attribute fingerprints can be stored on disk in a columnar format so that at query time, only the relevant predicates need to be read.

{\bf Small values:}
The attribute values themselves are often stored as small integers. One optimization is to hash only values $\geq 2^{|\kappa|}$. This way, all small values can be represented exactly.

{\bf Attribute compression:}
More accurate CCF's can be constructed using a two-stage process. First, construct a CCF with chaining using large attribute fingerprints. A compressed CCF can be constructed by mapping large attribute fingerprints to smaller ones while minimizing the number of collisions.

\subsection{Range queries}
\label{sec:range queries}
Conditional Cuckoo filters can be extended to support range predicates using standard techniques.
Given a column with a range predicate, one simple method is to bin the column into a small number of bins. A range predicate can then be converted into a small in-list. The disadvantages of this approach are that the binning process introduces error and that long ranges must check more bins. Since each bin that does not contain a true match can return a false positive, both can increase the FPR. 
Another method uses a standard approach of using a dyadic expansion over the range $[a_0, b_0]$ of the column. 
An item $x$ can be represented as a sequence of intervals $[a_1, b_1], \ldots [a_\eta, b_\eta]$
with exponentially decreasing lengths $b_{i+1} - a_{i+1} = (b_i - a_i)/2$ down to some final granularity
$b_\eta-a_\eta$. This requires $\eta$ insertions into a CCF for each item, and a range query likewise requires querying for the existence of up to $\eta$ intervals that cover the range.  We use the simpler binning approach in our experiments.

\vspace{-0.1cm}
\section{Experiments} \label{sec:experiments}
We consider two experiments. One on synthetic data examines the ability of the chaining procedure to store duplicate keys while obtaining a high load factor. The second examines the ability of the CCF to reduce output sizes on a join benchmark on real world data  \cite{kipf2019learned}.

\subsection{Multiset experiments}
\label{sec:multiset experiment}
Our experimental results show chaining dramatically improves the ability of a cuckoo filter to handle duplicate keys.
We simulated key frequencies using either a truncated Zipf-Mandelbrot distribution or a stream where every key has the same number of duplicates.

As the mean frequency of each key increases, figure \ref{fig:multiset dupes} shows that inserting into a regular cuckoo filter fails at lower load factors. Chaining allows the filter to maintain high load factors.
When the input has a Zipf-Mandelbrot distribution, the regular cuckoo filter fails extremely quickly. 

The setup for the multiset experiments are as follows.
 The chaining cuckoo hash parameter for the maximum number of duplicates per bucket pair is
 $d=3$, and the maximum number of bucket pairs for a key $L_{max} = \infty$ is uncapped.
 For each filter type and each setting for the average number of duplicates per key in the input data, we generate a dataset that is approximately 20\% larger than the capacity of the sketch and measure the number of items processed before the first failed insertion and the load factor at that point. A failed insertion here is the first time a unique key, attribute pair is not found in the sketch but fails to generate a new entry in the sketch. The order of items is randomly permuted.
 For the Zipf-Mandelbrot distribution with a mass function of the form $p(x) \propto (c + x)^{-\alpha}$, we fix the offset $c$ to be 2.7 and truncate the range to be in $[1,500]$. We vary $\alpha$ to obtain the desired average number of duplicates per key. We use the additional cycle detection and resolution method in these experiments. The results are averaged over 20 runs using random salts for the hash functions.
 
 \subsection{Multiset results}
 Figure \ref{fig:multiset dupes} shows the behavior of a regular multiset cuckoo filter and cuckoo filter with chaining as the number of duplicates per key is varied. The chained cuckoo filter is able to achieve roughly the same load factor regardless of the number of duplicates. In contrast, the plain cuckoo filter's ability to achieve a high load factor quickly decreases. For Zipf-Mandelbrot data, the plain cuckoo hash encounters very few items before it fails. 

 The cuckoo filter quickly encounters more copies of these keys than it can store.
  When the number of real duplicates per key matches $d$, the chained and plain cuckoo filters achieve similar load factors. However, the chained cuckoo filter has a slightly worse FPR as it inspects an extra bucket pair.

\begin{figure}[H]
  \vspace{-0.1cm}
  \centering
	\includegraphics[width=.8\textwidth]{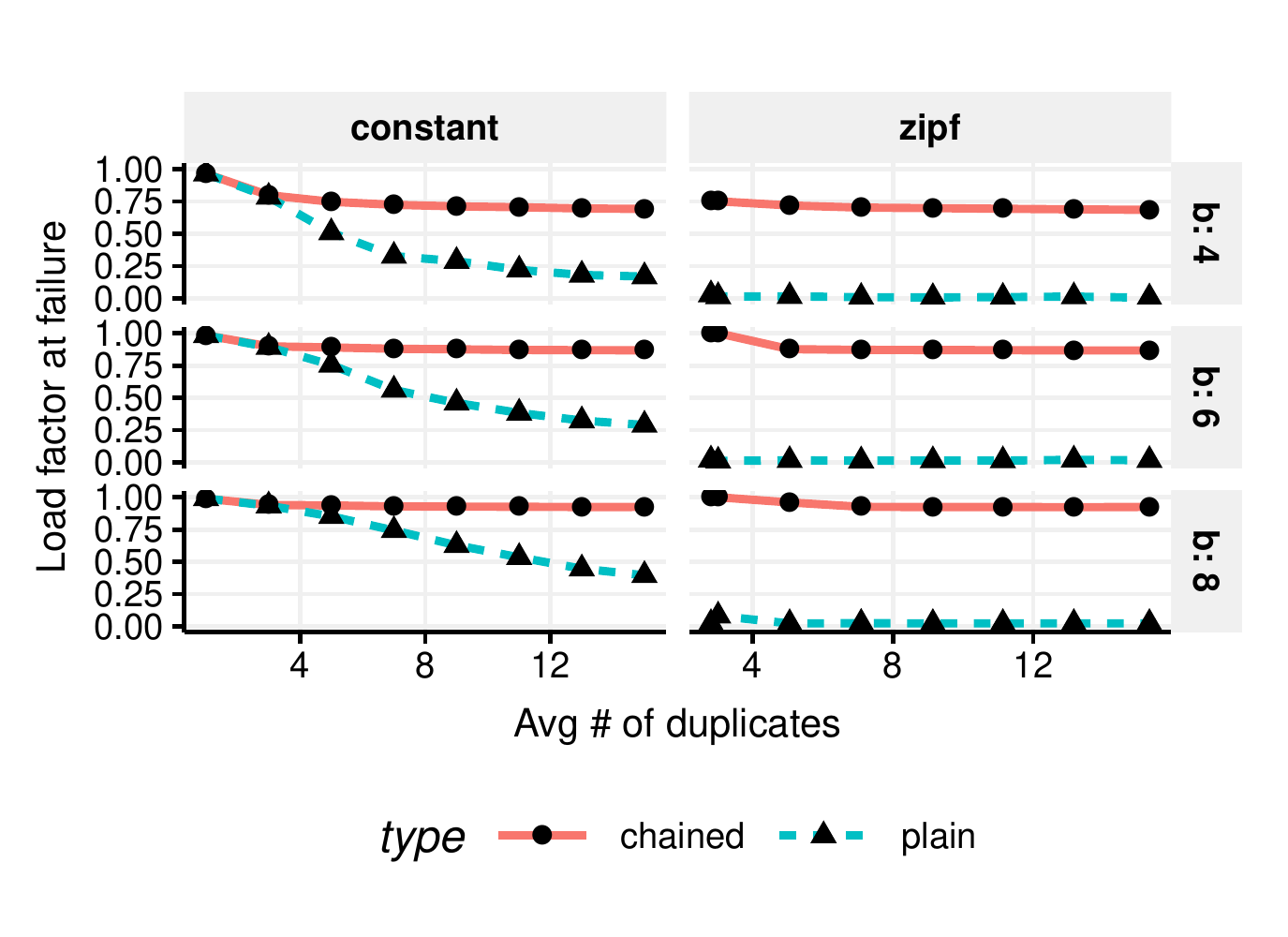}
	\caption{Chaining delays the first failed insertion and enables high load factors. A regular cuckoo filter achieves reasonable load factors only when the maximum number of replicates for a key is small. }
	\label{fig:multiset dupes}
  \vspace{-0.1cm}
\end{figure}
 
 Figure \ref{fig:multiset fpr/maxdupes} shows the efficiency of the chained cuckoo filter. We define the bit efficiency of the filter to 
 be 
 \begin{align}
 	\mathit{Efficiency} &\defn \frac{\mathrm{sketch\ size\ in\ bits}}{n \log_2 1/\rho} 
 \end{align}
where $n$ is the total number of keys inserted and $\rho$ is the FPR.
When all keys are distinct, a sketch with an efficiency of 1 cannot be improved since it matches the information theoretic lower bound given by the denominator. 
An optimized chained cuckoo filter obtained a bit efficiency of $\approx 1.93$ when all keys have the same number of duplicates $> d$.
In comparison, a Bloom filter has a bit efficiency of $\approx 1/\log 2 \approx 1.44$.
As shown in figure \ref{fig:multiset dupes}, a cuckoo filter can have arbitrarily bad bit efficiency when encoding multisets.
For sets, a cuckoo filter with the semi-sorting optimization has a bit efficiency 
of $\approx \beta^{-1} + 2 \beta^{-1} / \log_2 (1/\rho) \approx 1.37$ with a load factor of $95\%$ and FPR of $\rho = 1\%$.
Without the semi-sorting optimization, it has a bit efficiency of $\approx 1.53$ for the same FPR.
Thus, in the presence of many keys with duplicates, the efficiency decreases by a modest amount.

\begin{figure}[H]
  \vspace{-0.1cm}
  \centering
	\includegraphics[width=.8\textwidth]{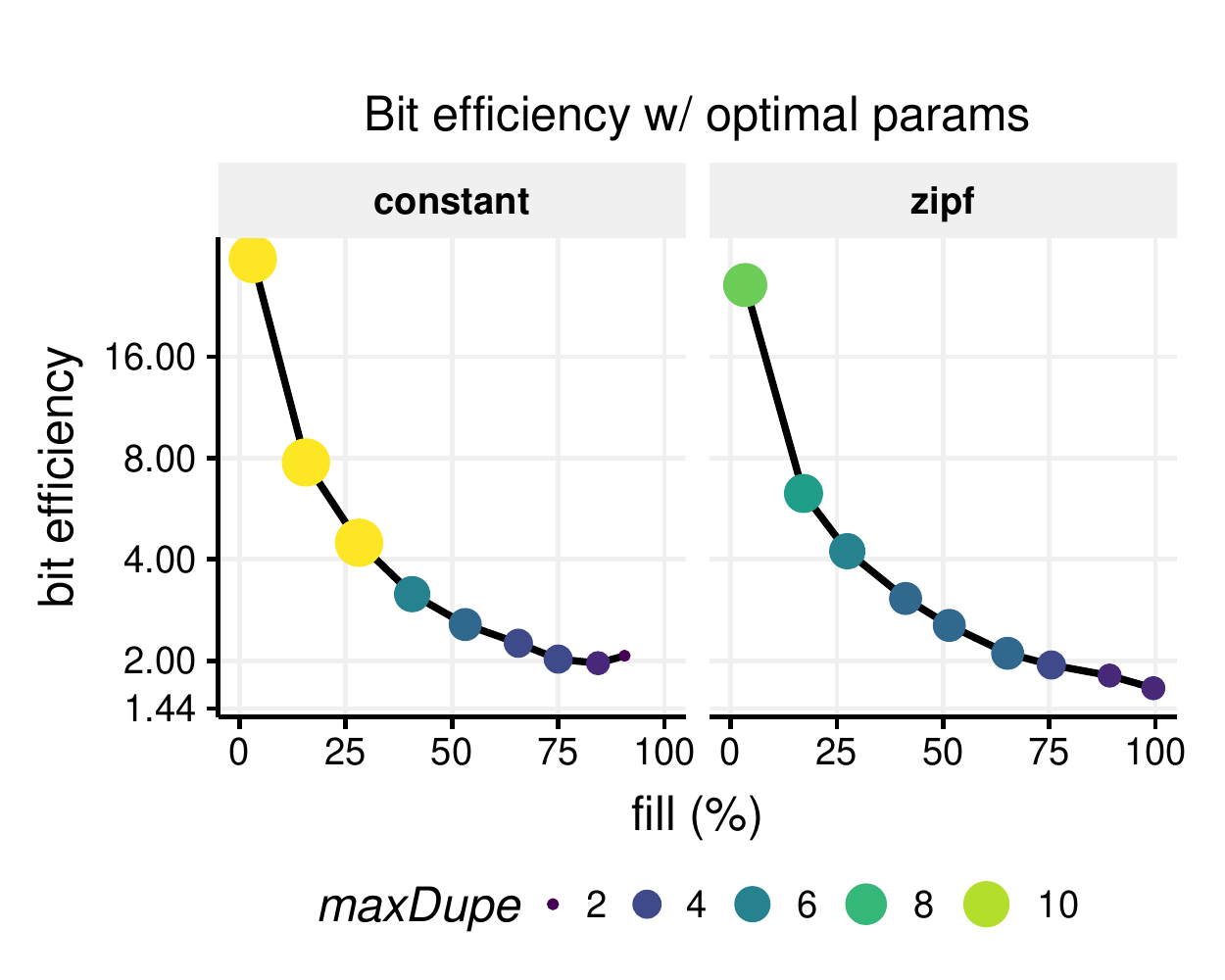}
	\caption{Higher load factors with small values for $d = $ maxDupe generally make better use of space. 
		While a very low value of 2 duplicates can achieve higher load factors, it can result in higher FPRs which make the sketch less efficient. }
	\label{fig:multiset fpr/maxdupes}
  \vspace{-0.1cm}
\end{figure}

\subsection{JOB-light experiments}
Our second set of experiments evaluate the efficacy of using CCF's for join processing on a real world dataset.
We evaluated the CCF's ability to reduce the output of scan operators using the JOB-light
workload~\cite{Kipf2018LearnedCE} for varying sizes of CCF's. This workload is derived from the Join Order Benchmark
(JOB)~\cite{leis2015good} that is used to evaluate the quality of query
optimizers. 
For a distributed system, the reduction factor measures how proportion of tuples are sent over the network, or in non-distributed hash joins, how much smaller the hash table sizes are. The CCF size is the space cost needed to achieve this reduction. This allows us to perform a comparison against the naive approach of simply building a hash table for each table without applying CCF pre-filtering.

The JOB-light workload consists of 70 queries, each joining up to 5 large
tables from the IMDB dataset.
Among these 70 queries are 237
instances of tables that qualify for matching join keys and predicates in at least one CCF,
effectively a semijoin reducer \cite{bernstein1981query}.
This reduction in the number of rows by semijoin(s) 
forms the basis for the following analysis based on a row reduction factor.

Given a query and a table in that query, we wish to determine the minimum sized output that a scan operator on that table can provide and compare it to an operator that only applies the predicates on that table.
The minimum size output for the scan operator is produced by converting joins of this base table to other tables to semijoins, which only check if the key exists in the other tables after applying predicates.
To apply a CCF with predicates, for a row with key $k$, each relevant CCF is queried for $(k, \pcal)$ where $\pcal$ is the predicate from the query.
We define the Reduction Factor (RF) to be
\begin{align}
  \mathit{Reduction\ Factor} &\defn \frac{M_{\mathit{semijoin}}}{M_{\mathit{predicate}}}
\end{align}
where $M_{\mathit{semijoin}}$ is the number of base table rows that both match the given predicates on that table and all other CCF's.
The value $M_{\mathit{predicate}}$ is the total number of rows in the base table that match the given predicates with no additional information from other tables.
The concept of reduction factor is related to that of predicate selectivity.
When reduction factor (selectivity) is 0.0 then no rows are selected and when reduction factor (selectivity) is 1.0
then all rows are selected.

In this workload, each query involves 2 to 5 of the 6  tables listed in table \ref{tbl:workload tables}, and all joins are on the movie identifier. 
Thus, between 1 to 4 CCF's may be applied to each query, given one CCF per table.
Two tables, \texttt{movie\_companies} and \texttt{title}, each contain two predicate columns, thus providing an opportunity to
evaluate the workload using a combination of single-attribute and multi-attribute CCF's.

\begin{table}[H]
  \centering
  {\begin{tabular}{l|r|l|r}
    & Number  & Predicate & Column \\
    Table & of Rows & Column    & Cardinality \\
    \hline
    cast\_info & 36,244,344 & role\_id & 11 \\
    movie\_companies & 2,609,129 & company\_id & \em{234,997} \\
    movie\_companies & 2,609,129 & company\_type\_id & 2 \\
    movie\_info & 14,835,720 & info\_type\_id & \em{71} \\
    movie\_info\_idx & 1,380,035 & info\_type\_id & 5 \\
    movie\_keyword & 4,523,930 & keyword\_id & \em{134,170} \\
    title & 2,528,312 & kind\_id & 6 \\
    title & 2,528,312 & production\_year & \em{132}
  \end{tabular}}
  \caption{Summary of tables and predicates used in JOB-light workload. "High" cardinalities emphasized.}
  \label{tbl:workload tables}
  \vspace{-0.3cm}
\end{table}

While most predicates are equality predicates, 55 JOB-light queries have inequality predicates on \texttt{title.production\_year}.
Because \texttt{production\_year} is an integer ranging from 1880 to 2019, we applied the simple binning technique in section \ref{sec:range queries} and mapped the 132 values to 16 roughly equal-sized intervals.
The inequality predicates were then converted to an in-list. In the cases where the scan operation was on the \texttt{title} table, we omitted this binning since the predicate could be evaluated directly.

The relevant IMDB data for the JOB-light workload is summarized in tables \ref{tbl:workload tables} and \ref{tbl:workload duplicates}.
These include the predicate columns and their cardinalities.
When using 4 bit attribute fingerprints, 4 of the predicates can be considered "high"
cardinality as the number of possible attribute values exceeds the 16 possible values for
attribute fingerprints. Even when using 8 bit attribute fingerprints, 2 of the predicates significantly exceed 256
possible values for attribute fingerprints.
Table \ref{tbl:workload duplicates} shows the number of duplicate predicate attribute values per join key.
This affects both the sizing of the sketches as well as the FPR.
The worst case is \texttt{movie\_keyword.keyword\_id}, which may have up to 539 distinct attribute
values for a single \texttt{movie\_id} join key.
While the \textit{Avg Dupes} reasonably vary from 1.00 to 9.48 distinct duplicates per join key,
CCF's must handle the worst case behavior of \textit{Max Dupes} which varies from 1 to 539.

\begin{table}[H]
  \centering
  {\begin{tabular}{l|l|l|r|r}
    & & Predicate & Avg & Max \\
    Table & Join Key & Column & Dupes & Dupes \\
    \hline
    cast\_info & movie\_id & role\_id & 4.70 & 11 \\
    movie\_companies & movie\_id & company\_id & 2.14 & \em{87} \\
    movie\_companies & movie\_id & company\_type\_id & 1.54 & 2 \\
    movie\_info & movie\_id & info\_type\_id & 4.17 & \em{68} \\
    movie\_info\_idx & movie\_id & info\_type\_id & 3.00 & 4 \\
    movie\_keyword & movie\_id & keyword\_id & 9.48 & \em{539} \\
    title & id & kind\_id & 1.00 & 1 \\
    title & id & production\_year & 1.00 & 1 \\
  \end{tabular}}

  \caption{Summary of average and maximum number of distinct, duplicate predicate attribute values per key. "High" Max Dupes emphasized.}
  \label{tbl:workload duplicates}
    \vspace{-0.3cm}
\end{table}

\subsection{JOB-light experiment setup}
We evaluated all four CCF methods:
Plain (regular cuckoo filter allowing duplicate keys);
Chained (CCF w/ chaining);
Bloom (CCF w/ Bloom); and
Mixed (CCF w/ Bloom conversion);
for a range of attribute sizes, fingerprint sizes, and Bloom
attribute sketch sizes.  Additionally, we evaluated these CCF's with and without
predicates on the attribute values of the tables being joined.

The following filter parameters were evaluated: 
attribute fingerprint sizes of $|\alpha| = 4$ or  $8$ bits,
fingerprint sizes of $|\kappa| = 7, 8,$ or $12$ bits, and
Bloom filter sizes ranging from 4 to 24. The
number of hash functions used in the Bloom filter was either fixed at 2, or was optimized to achieve the lowest FPR under the assumption that 2 attribute vectors are inserted per key. We found the latter setting resulted in uniformly worse FPR's and omit their results from the rest of this analysis.
In the case of Bloom conversion, this was $d+1$ instead. 
The maximum number of duplicate key fingerprints per bucket pair was always set to $d=3$.

The number of buckets and bucket size were independently chosen for each filter based on the 
analysis in section \ref{sec:param choice}. 
Given the predicted number of entries, we find the smallest bucket size which would both result in an acceptable load factor and have high likelihood of successfully inserting all input rows based on the multiset experiment. That is all runs for a  given bucket size failed at a higher load factor than the predicted load factor in the multiset experiment.
Note that the predicted number of entries needed can be estimated from the data using a bottom-k  \cite{cohen2007bottomk}  or two-level \cite{chen2017two} sampling scheme.

\subsection{JOB-light results}
We  compare how the CCF methods perform versus the theoretically best possible RF which makes full use of all predicate information and against
the best existing baseline of a Cuckoo Filter, which throws away information about predicates.
The best possible RF is the \textit{Exact Semijoin reduction factor} where no false positives are emitted,
and best existing baseline is the \textit{Cuckoo Filter reduction factor}.
We also examined the effect of changing the size of fingerprints and attribute sketches on the accuracy of CCF's.

{\bf Large filters:}
Figure \ref{fig:optimal largewithpredicates} plots the reduction factor on the y-axis for each of the 237
instances of table join key and predicate matches in the 70 workload queries. The parameters of all
filters are the same, having 8-bit attributes, 12-bit fingerprints, 4 hash functions for Bloom
filters.
These filters as "large" due to the choice of 8 and 12 bits for attribute and fingerprint size
respectively, as well as specifying a large Bloom filter sketch.
The ordering of these tables on the x-axis is in increasing order of the Exact Semijoin reduction factor;
therefore, the reduction factor of all
filters should be to the left and above the Exact Semijoin RF line.  

For large filters, the reduction factors of all filter methods are fairly close to Exact Semijoin with a few outliers.

Figure \ref{fig:keysonly largewithpredicates} also uses large filters, but has a different baseline than in \ref{fig:optimal largewithpredicates}.
Here the baseline is based on using large Bloom filters, but only looking up the filters' join keys, ignoring any other predicates.
This behavior is analogous to a regular Cuckoo filters rather than a CCF and represents the current state-of-the-art of pre-built filters.
The reduction factor of all filters should be to the right and below the Cuckoo Filter baseline.
\textit{CCF reduction factors are substantially better than current state-of-the-art pre-built filters.}
In many cases, where the Cuckoo Filter reduction factor is 1.0, meaning no reduction at all, the CCF RF's are in the range of $0.05 - 0.20$.

{\bf Small filters:}
Now we consider "small" filters,
using 4-bit attributes, 7-bit fingerprints, and 2 hash functions for Bloom, reducing filter size by more than half.
Figures \ref{fig:optimal smallwithpredicates} and \ref{fig:keysonly smallwithpredicates}
show the reduction factors by increasing Exact Semijoin and Cuckoo Filter RF baselines respectively, for filters of smaller size.
Compared to large filters, the number of non-optimal reduction factors in \ref{fig:optimal smallwithpredicates} are more visible. 
The separation of Bloom CCF reduction factors from Mixed and Chained is particularly noticeable, while
small Mixed and Chained RF's are similar to large filters.
\textit{Even small CCF's are substantially better than current state-of-the-art filters.}

\newpage

\begin{figure*}[ht]
  \centering
    \subcaptionbox{Large vs Exact Semijoin\label{fig:optimal largewithpredicates}}
      {\includegraphics[width=.44\textwidth]{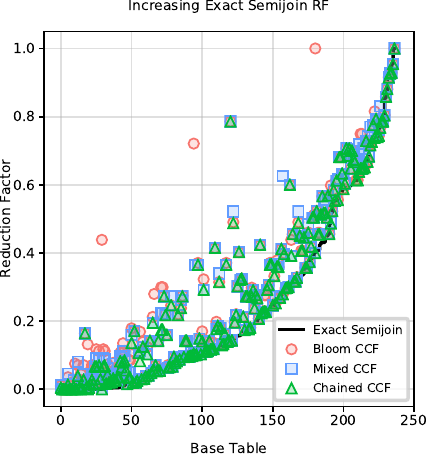}\hspace{-0.1cm}}
    \subcaptionbox{Large vs Cuckoo Filter\label{fig:keysonly largewithpredicates}}
      {\includegraphics[width=.44\textwidth]{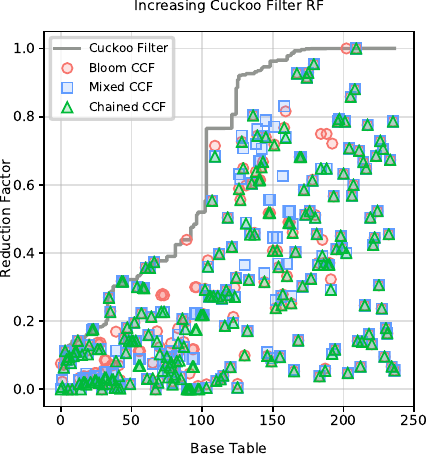}\hspace{-0.1cm}}
    \subcaptionbox{Small vs Exact Semijoin\label{fig:optimal smallwithpredicates}}
      {\includegraphics[width=.44\textwidth]{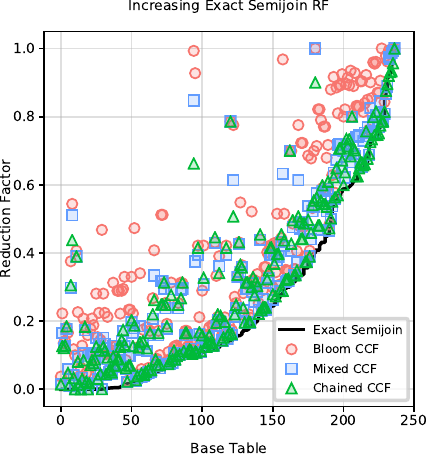}\hspace{-0.1cm}}
    \subcaptionbox{Small vs Cuckoo Filter\label{fig:keysonly smallwithpredicates}}
      {\includegraphics[width=.44\textwidth]{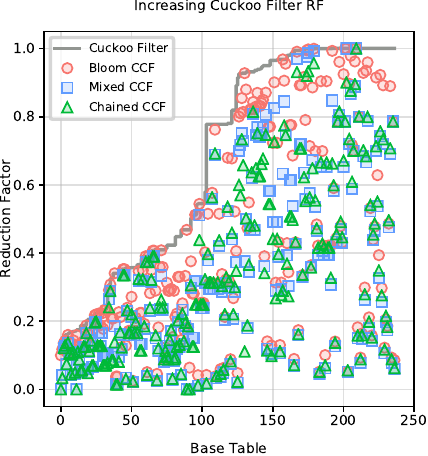}\hspace{-0.0cm}}
	\caption{Reduction factors for CCF's versus the Exact Semijoin reduction factor baseline or Cuckoo Filter baseline. Using large filters for \ref{fig:optimal largewithpredicates} and \ref{fig:keysonly largewithpredicates} and small filters for \ref{fig:optimal smallwithpredicates} and \ref{fig:keysonly smallwithpredicates}. Small CCF's exhibit an increased FPR. The reduction factors of CCF's are much improved over Cuckoo Filters.
  }\label{fig:combined withandwithoutpredicates}
\end{figure*}

{\bf Plain filters:}
Note that none of these figures have results for Plain CCF filters as they did not result in reasonably sized filters.
The smallest Plain filter was larger than every CCF and had an inefficient load factor of 35\%.
Larger attribute fingerprints result in insertion failures for any reasonable filter parameters because the number of distinct attribute values is too large.
For example, as shown in table \ref{tbl:workload duplicates}, \texttt{movie\_keyword.keyword\_id} has 539 distinct duplicates which would require a minimum bucket size of 270.

\subsection{JOB-light aggregate results}

On aggregate, the reduction factor over all table scans was $\approx 0.28$ using a CCF with chaining and "small" sketches. In contrast, using regular cuckoo filters with no predicate information resulted in a reduction factor of $\approx 0.68$. 
The best possible reduction factor obtained from performing an exact semi-join was $0.20$. Furthermore,
half of the difference in reduction factors between the CCF and exact semi-join is explained by the binning on the range predictate on \texttt{production\_year}
as seen in figures \ref{fig:scatterplot withbinnedyear} and \ref{fig:overall RF}.
If an exact semi-join is performed on data with binned \texttt{production\_year}, the optimal reduction factor is $0.24$.

\begin{figure}[H]
  \centering
	\begin{tabular}{cc}
		\hspace{-0.14cm}
	\includegraphics[width=.44\textwidth]{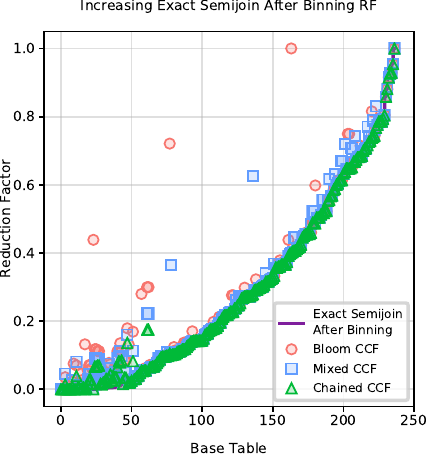}&
	\hspace{-0.48cm}
		\includegraphics[width=.44\textwidth]{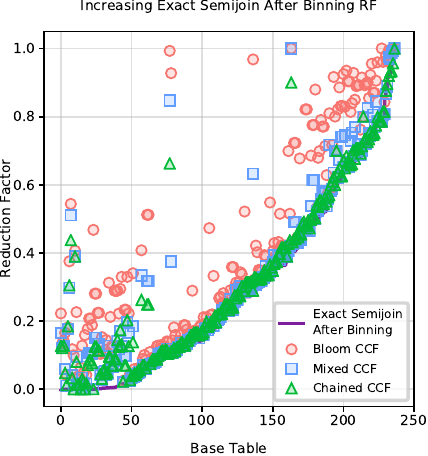}
	\end{tabular}
  \caption{Reduction factors for large (on the left) and small (on the right) CCF's versus the exact semijoin baseline after binning \texttt{title.production\_year}. The FPR's of both large and small CCF's are noticeably less than those in \ref{fig:optimal largewithpredicates} and \ref{fig:optimal smallwithpredicates} respectively.}
\label{fig:scatterplot withbinnedyear}
\end{figure}

Furthermore, using the largest sized CCF, which was a CCF with chaining, 12 bit key fingerprints, and 8 bit attribute fingerprints, the FPR was just $0.8\%$ relative to a semi-join with binned with \texttt{production\_year}, and the reduction factor was $0.245$.  
The FPR including errors due to binning was $6.1\%$.

Figure \ref{fig:RF by num joins} shows that the benefits of CCF's are compounded as more joins are added.
Figure \ref{fig:size reduction per column} shows the size of a CCF relative to the raw data. Each CCF represents a movie id and the given predicate column. The gains from each CCF can varying significantly based on the underlying data and the number of duplicate keys.

\vspace{-0.1cm}
\begin{figure}[H]
  \centering
	\includegraphics[width=.8\textwidth]{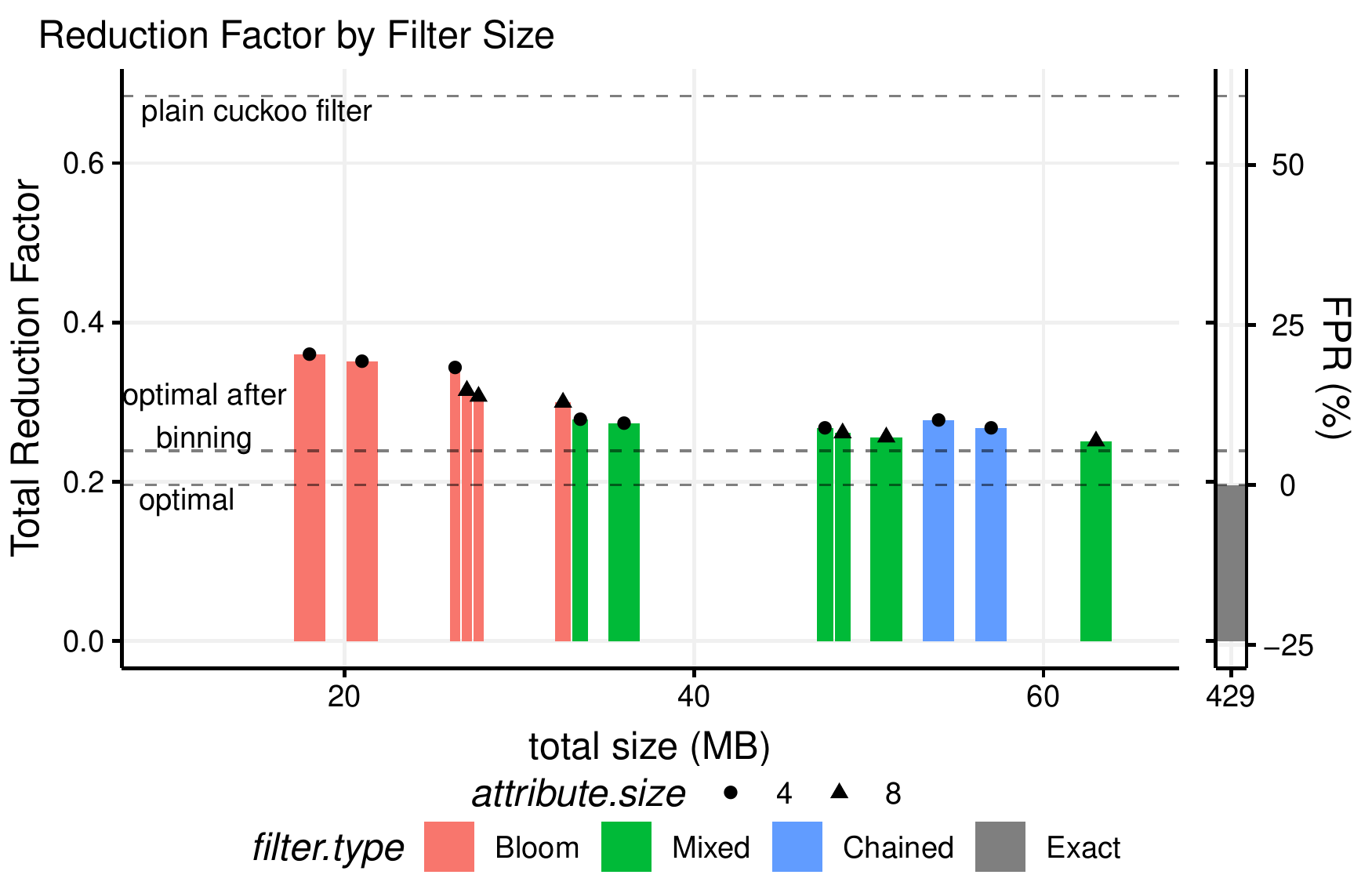}
	\caption{Overall RF and FPR  by filter type and size. CCF's obtain significantly better reduction factors while using an order of magnitude less space than a hash table performing exact membership testing. Bloom attribute sketches resulted in the smallest sizes. Mixed attribute sketches using Bloom conversion achieved the smallest FPR for a given amount of space. A modestly sized CCF that is $1/7^{th}$ the size of a raw hash table contributes a negligible number of false positives other than those introduced from binning \texttt{production\_year}.}
	\label{fig:overall RF}
\end{figure}

\begin{figure}[H]
	\vspace{-0.1cm}
  \centering
	\includegraphics[width=.8\textwidth]{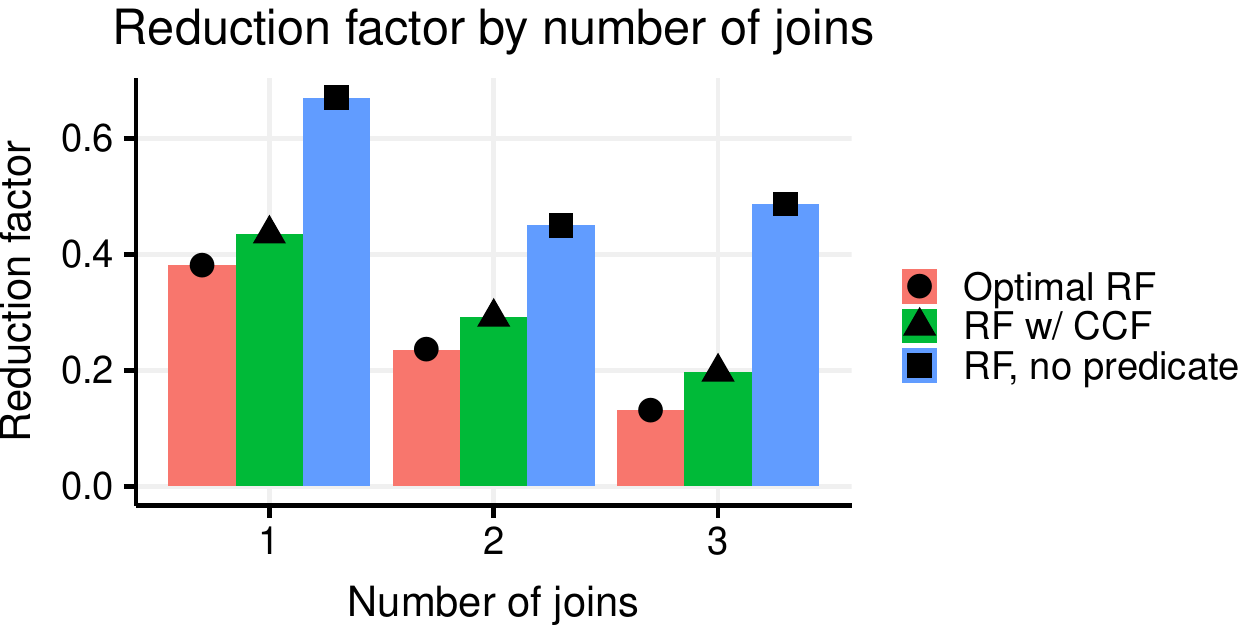}
	\caption{Using CCF's with predicates results in multiplicative benefits in reducing output sizes.}
	\label{fig:RF by num joins}
	\vspace{-0.1cm}
\end{figure}

\begin{figure}[H]
	\vspace{-0.4cm}
  \centering
 \includegraphics[width=.8\textwidth]{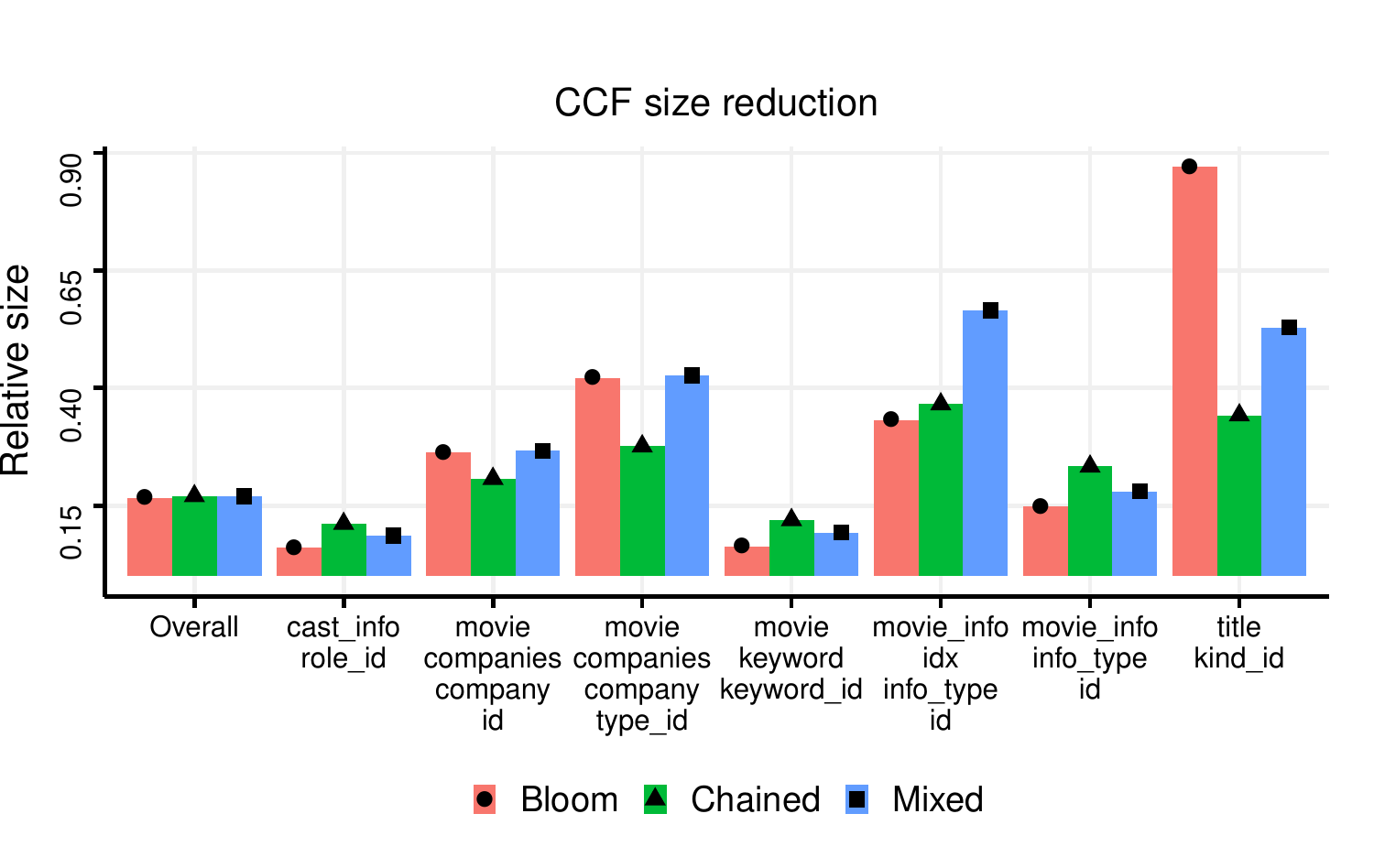}
 \caption{Given CCF's of equal size, each CCF's size relative to its underlying data varies significantly. Bloom filters yield larger size reductions for tables with many duplicated keys while  chaining  yields larger reductions on tables with unique keys. }
 	\vspace{-0.2cm}
 	\label{fig:size reduction per column}
\end{figure}

\subsection{CCF Size}
We did not find that one method clearly outperformed others at any given size. However, Bloom filter attribute sketches were the  ones that could produce very small sketches. The smallest set of sketches required 18.5 megabytes. 
By comparison, if keys and high cardinality attributes are stored using 32 bits and low cardinality attributes are stored with 8 bits, then the raw data used in the sketches requires 322 megabytes of space. A open addressing hash table would require 429 megabytes to store these if it could achieve a 75\% load factor. The much smaller sizes of the CCF are due both to sketching keys and attributes as well as due to elimination of duplicate keys.

Figure \ref{fig:overall RF} shows the reduction factor as a function of the total size of all sketches for various parameter settings. The maximum sketch size was restricted to 60 megabytes. Bloom filter attribute sketches yield the smallest possible sketches since they store no duplicates, but they also yield the highest FPR of up to $20.4\%$. The mixed attribute sketches which use attribute fingerprint vectors and switch to Bloom filters when there are too many duplicates were able to retain most of the benefits of attribute fingerprint vectors with significantly less space.

For the range of parameters we considered, given two sketches with the same size, allocating more space to the attribute sketch yielded a smaller FPR than allocating more space to the key fingerprint size. This can be seen in figure \ref{fig:overall RF} in cases where there are adjacent, identically colored half-width bars but different marks on the bar.

Figure \ref{fig:overall RF} shows the space versus reduction factor tradeoff for different CCF types.  
At 35 megabytes, the reduction factor of a CCF with mixed attributed sketches is within 10\% of the optimal and within 5\% 
of the best a CCF can do with binned \texttt{production\_year}.
While larger CCFs further improve the FPR, the practical benefit in the reduction factor is minimal. A 35 megabyte set of CCFs represents an over $10\times$ reduction in size in memory.
If string based columns are included, which are very common in real world databases \cite{vogelsgesang2018get}, we expect the space savings to be much greater.

\subsection{Run-time performance}
While we did not optimize our single threaded C++ implementation of CCF's, all filter methods could process 1 million matches per second.
(Measurements were taken using a single Intel(R) Xeon(R) CPU E5-2630 v4 @ 2.20GHz core running CentOS Linux 7.)
 We used the Jenkins lookup3 hash function \cite{jenkins2006lookup3} that is also used by the original cuckoo filter paper \cite{fan2014cuckoo}. For attribute fingerprints, we used the small value optimization given in section \ref{sec:optimizations}.

\section{Discussion and future work} \label{sec:discussion}
While we focus on the application of CCF's to join processing, the sketch itself can be seen as a sketch of the entire input table with a hash based index on the key. Furthermore, the chaining technique can also be used to allow regular cuckoo hash tables, which store the full key, to store duplicates. Thus, we believe the sketch and its methods have more general applications beyond join processing. 

We have made the sketch much more robust to duplicated keys so that the chance of successfully inserting all rows is much more predictable if the predicted number of filled entries can be estimated. However, to predict the number of filled entries requires another sketch to be computed first.  Future work to improve this method, as well as many other approximate set membership sketches, includes enabling dynamic adjustment of the size of the sketch.

Furthermore, while we present empirical evidence that, on data containing duplicate keys, the CCF with chaining can achieve a load factor comparable to that of a regular cuckoo filter acting on data with no duplicates, we do not have a theoretical proof that this is always the case.

\section{Conclusions} \label{sec:conclusions}
We introduce conditional cuckoo filters, a new sketch for approximate set membership queries which enables equality predicates to be added to queries. This yields at least two significant advantages in join processing. First, it enables filters that are specific to the predicate to be applied to both build and probe sides of a join, not just the probe side. This increases the number of cases where the data structures created on the build side fits into main memory. Second, it enables predicate pushdown from one table to all other tables in the transitive closure of the join graph. 
 
We propose, analyze, and evaluate multiple variations of CCF sketches. In particular, we extend cuckoo hash tables using a chaining technique that makes it robust to duplicated keys and allows high load factors to be achieved. All variations reduced the number of rows emitted by a scan operator to close to the optimal number on the workload and did so with substantial space savings. This represents a significant improvement over existing filters that do not support predicates. The properties of the sketches are analyzed which allows practitioners to predict the performance of the sketches and to choose appropriate parameters for them.

\bibliographystyle{abbrv}
\bibliography{ling2,rick}

\end{document}